%% file: CompoundMIMOBCC-CR.tex
\documentclass[10pt,conference]{IEEEtran}
\usepackage{amsmath,amssymb,eucal,graphicx}
\usepackage{epsfig}
\usepackage{exscale}

\newtheorem{definition}{Definition}%[section]
\newtheorem{lemma}{Lemma}%[section]
\newtheorem{cor}{Corollary}[section]

\newtheorem{assumption}{Assumption}[section]

\newtheorem{theorem}{Theorem}
\input{macros}

\DeclareFontFamily{U}{cmfi}{}
\DeclareFontShape{U}{cmfi}{m}{n}{ <-> cmfi10 }{}
\DeclareSymbolFont{CMFI}{U}{cmfi}{m}{n}

\begin{document}

\vspace{-3cm}
\title{On the Compound MIMO Broadcast Channels with Confidential Messages}

\author{\authorblockN{Mari Kobayashi}
\authorblockA{SUPELEC \\
Gif-sur-Yvette, France\\
%Email: {\tt mari.kobayashi@supelec.fr}\\
}
\and
\authorblockN{Yingbin Liang}
\authorblockA{University of Hawaii\\
Honolulu, HI, USA\\
%Email: {\tt yingbinl@hawaii.edu}\\
}
\and
\authorblockN{Shlomo Shamai (Shitz)}
\authorblockA{ Technion-Israel Institute of Technology\\
Haifa, Israel\\ %Email: {\tt }\\
}
\and
\authorblockN{M\'erouane Debbah}
\authorblockA{SUPELEC \\
Gif-sur-Yvette, France\\ %Email: {\tt }\\
}
}

\maketitle

\thispagestyle{plain}

\begin{abstract}
We study the compound multi-input multi-output (MIMO) broadcast
channel with confidential messages (BCC), where one transmitter
sends a common message to two receivers and two confidential
messages respectively to each receiver. The channel state may take
one of a finite set of states, and the transmitter knows the state
set but does not know the realization of the state. We study
achievable rates with perfect secrecy in the high SNR regime by
characterizing an achievable secrecy degree of freedom (s.d.o.f.)
region for two models, the Gaussian MIMO-BCC and the ergodic
fading multi-input single-output (MISO)-BCC without a common
message. We show that by exploiting an additional temporal
dimension due to state variation in the ergodic fading model, the
achievable s.d.o.f. region can be significantly improved compared
to the Gaussian model with a constant state, although at the price
of a larger delay.
\end{abstract}

%%%%%%%%%%%%%%%%%%%%%%%%%%%%%%%%%%%%%%%%%%%%%%%%%%%%%%%%%%
\section{Introduction}
% imperfect CSIT in MIMO-BCC and problem setting
In most practical scenarios, perfect channel state information at
transmitter (CSIT) may not be available due to time-varying nature
of wireless channels (in particular for fast fading channels) and
limited resources for channel estimation. However, many wireless
applications must guarantee secure and reliable communication in
the presence of the channel uncertainty. In this paper, we
consider such a scenario in the context of the multi-input
multi-output (MIMO) broadcast channel, in which a transmitter
equipped with multi-antennas wishes to send one common message to
two receivers and two confidential messages respectively to the
two receivers. The channel uncertainty at the transmitter is
modeled as a compound channel, i.e., the channel to two receivers
may take one state from a finite set of states. The transmitter
knows the state set, but does not know the realization of the
channel state. The transmitter needs to send all messages reliably
while keeping each confidential message perfectly secret from the
non-intended receiver, no matter which channel state occurs.

We note that the compound MIMO broadcast channel with confidential messages (BCC) is not yet fully understood.
This can be expected from two special cases studied in
\cite{weingarten2007cmb} and \cite{LiuPoorIT09}. One the one hand,
it is well known that without secrecy constraints the capacity
region of the MIMO-BC under general CSIT is unknown. Moreover,
even the d.o.f. region of the compound MIMO-BC is not fully known
despite the recent progress \cite{weingarten2007cmb}. On the other
hand, although the secrecy capacity region of the two-user
MISO-BCC has recently been characterized \cite{LiuPoorIT09}, the
secrecy capacity of a general MIMO-BCC remains open.

In this paper, we study achievable secrecy degree of freedom
(s.d.o.f.) regions of the MIMO-BCC, which characterize the
behavior of an achievable secrecy rate region in the high
signal-to-noise (SNR)
regime.  %\textbf{[shlomo you want to cite your work here ?]}.
% our models, contribution
We consider two compound MIMO-BCC models. The first model is the
Gaussian compound MIMO-BCC, in which the channel remains in the
same state during the entire transmission. We assume that each
terminal is equipped with multiple antennas and the transmitter
sends one common message as well as two confidential messages to
two receivers. We propose a beamforming scheme to obtain an
achievable s.d.o.f., and characterize the impact of the number of
antennas and the number of channel states on this region. We show
that with $M$ transmit antennas, $N_k$ receive antennas and $J_k$
states for $k=1,2$, a positive s.d.o.f. is ensured to both
receivers only if the number of transmit antennas is sufficiently
large, i.e. $M>\max(J_1N_1,J_2N_2)$.

The second model we study is the ergodic fading compound
multi-input single-output (MISO)-BCC, where the channel remains in
one state for a block duration and then changes independently from
one block to another. We model the channel state at each block as
a set of random variables uniformly distributed over a finite set.
Applying the variable-rate transmission strategy proposed for the
ergodic fading wiretap channel with partial CSIT
\cite{gopala2006scf}, we characterize an achievable s.d.o.f.
region. It is shown that time variation of the channel (which
introduces an additional temporal dimension) enables to improve
the s.d.o.f. region compared to the Gaussian model with constant
channel state, although the second model applies only to
delay-tolerant applications.
%It is found that the secrecy constraint further decreases the d.o.f. region compared to the classical compound %MIMO-BC \cite{weingarten2007cmb}.

% relevant works on compound MIMO-BC, MIMO-BCC,
We note that the compound MIMO-BCC yields a number of previously
studied models as special cases. For the special case of perfect
CSIT, the secrecy capacity region of the two-user MISO-BCC has
been recently characterized in \cite{LiuPoorIT09}. A more general
two-user MIMO-BCC is considered in \cite{Liangisita08}, where the
secrecy capacity region of the MIMO-BCC with one common message
and one confidential message is characterized. For the
frequency-selective BCC modeled as a special Toeplitz structure of
the MIMO-BCC, the s.d.o.f. region is analyzed in
\cite{VandermondeSubmitted}. All above studies do not address the
compound nature of the channel. For the special case of only one
confidential message, the capacity of the degraded MIMO compound
wiretap channel is characterized and an achievable s.d.o.f. of the
MIMO compound wiretap channel is derived in
\cite{liangallerton07}. The s.d.o.f. of the compound wiretap
parallel channels is considered in \cite{liangpimrc08,liu2008scc}.

% organization
The paper is organized as follows. In Sections \ref{sec:deterministic} and \ref{sec:ergodic}, we study the Gaussian MIMO-BCC and the ergodic fading MISO-BCC, respectively. Section \ref{sec:conclusions} concludes the paper.

In this paper, we adopt the following notations. We let
$[x]_+=\max\{0,x\}$ and $C(x)=\log(1+x)$. We use $\xv^n$ to denote
the sequence $(\xv_1,\dots,\xv_n)$, and use $u,v,w,\xv,\yv$ to
denote the realization of the random variables $U,V,W,X,Y$. We use
$|\Am|, \Am^H, \trace(\Am)$ to denote the determinant, the
hermitian transpose, and the trace of a matrix $\Am$,
respectively.

%%%%%%%%%%%%%%%%%%%%%%%%%%%%%%%%%%%%%%%%%%%%%%%%%%%%%%%%%%
\section{Gaussian compound MIMO-BCC}\label{sec:deterministic}

\subsection{Model and Definitions}\label{sec:gaussmodel}

We consider the MIMO-BCC, where the transmitter sends the
confidential messages $W_1,W_2$ to receivers 1 and 2 as well as a
common message $W_0$ to both receivers. The transmitter, and
receivers 1 and 2 are equipped with $M, N_1,N_2$ antennas,
respectively. The transmitter knows a discrete set of possible
channel states, and each receiver has perfect CSI.
%This model can be also interpreted as the multicast channel with 2 groups of users ($J_1$ and $J_2$ users).
The channel output of receiver $k$ in state $j$ at each channel use is given by
\begin{eqnarray}
 \yv_{k,j} &=& \Hm_k^j \xv + \nuv_{k}^j,\;\; j=1,\dots,J_k, \; k=1,2
\end{eqnarray}
where $J_k$ denotes the number of possible channel states of
receiver $k$, $\Hm_k^j\in \CC^{N_k\times M}$ is the channel matrix
of user $k$ in state $j$, $\nuv_{k}^j\sim\Nc_{\Cc}(\zerov,\Id)$ is
an additive white Gaussian noise (AWGN) and is independent and
identically distributed (i.i.d.) over $k$ and $j$, and the
covariance $\Sm_x$ of the input vector $\xv$ satisfies the power
constraint $\trace(\Sm_x)\leq P$. For the channel matrices, we
have the following assumption.
\begin{assumption}\label{assumption:rank}
Any $M$ rows taken from the matrices $ \Hm_1^1, \ldots,
\Hm_1^{J_1}, \Hm_2^1, \ldots, \Hm_2^{J_2}$ has rank $M$.
\end{assumption}

%We consider the following code construction.
\begin{definition} A $(2^{nR_0},2^{nR_1}, 2^{n R_2}, n)$ block code for the Gaussian compound MIMO-BCC consists of following:
\begin{itemize}
  \item Three message sets : $\Wc_i=\{1,\dots,2^{nR_i}\}$ and $W_i$ is uniformly distributed over $\Wc_i$ for $i=0,1,2$.
  \item A stochastic encoder that maps a message set $(w_0,w_1,w_2)\in(\Wc_0,\Wc_1,\Wc_2)$ into a codeword $\xv^n$.
  \item Two decoders : decoder $k$ maps a received sequence $\yv_{k,j}^n$ into $(\hat{w}_0^{(k,j)},\hat{w_k}^{(j)})\in(\Wc_0,\Wc_k)$ for $j=1,\dots,J_k$ and $k=1,2$.
\end{itemize}
\end{definition}
A rate tuple $(R_0, R_1,R_2)$ is {\em achievable} if for any
$\epsilon>0 $, there exists a $(2^{nR_0},2^{nR_1}, 2^{n R_2}, n)$
block code such that the average error probability of receivers 1
and 2 at state $(j,l)$ satisfy
\begin{eqnarray*}
P_{e,1,j}^{(n)} \leq \epsilon, \;\;  P_{e,2,l}^{(n)} \leq \epsilon,
\end{eqnarray*}
and
\begin{eqnarray}\label{eq:SecurityConstraints}
 nR_1 -H(W_1|Y_{2,l}^{n}) \leq n\epsilon, \;\; nR_2 - H(W_2|Y_{1,j}^{n}) \leq n\epsilon
\end{eqnarray}
for any $j=1,\dots,J_1,l=1,\dots,J_2$. Note that
\eqref{eq:SecurityConstraints} requires perfect secrecy for the
confidential messages at the non-intended receiver.

We further define the degree of freedom (d.o.f.) of the common message and the secrecy degree of freedom of the confidential messages as
\[r_0 = \lim_{P\rightarrow \infty}\frac{R_0(P)}{\log (P)},  \;\;\;  r_k = \lim_{P\rightarrow \infty}\frac{R_k(P)}{\log(P)},\; k=1,2. \]

\subsection{Secrecy Degree of Freedom Region}
An achievable secrecy rate region for the discrete memoryless broadcast channel with one common and two confidential messages was given in \cite{xu2008cbb}. We can extend this result to the corresponding compound channel studied in this paper and obtain an achievable secrecy rate region given by
%as a union of all $(R_0,R_1,R_2)$ satisfying
\begin{eqnarray}\label{DeterministicRegion}
%  R_s = \cov \bigcup \{\Rm :
&& 0 \leq  R_0 \leq\min_{k,j}I(U;Y_{k,j})\\ \nonumber
 && 0\leq R_1  \leq  \min_{j,l} [I(V_1;Y_{1,j}|U)-I(V_1;Y_{2,l},V_2|U)] \\ \nonumber
&& 0\leq  R_2  \leq  \min_{j,l} [I(V_2;Y_{2,l}|U)-I(V_2;Y_{1,j},V_1|U)]
 %\}
\end{eqnarray}
over all possible joint distributions of $(U,V_1,V_2,X)$
satisfying
\begin{equation}
    U \rightarrow (V_1,V_2) \rightarrow X \rightarrow (Y_{1,j},Y_{2,l}), \forall
    j,l.
\end{equation}

Based on the preceding region, we obtain the following theorem on an achievable s.d.o.f. region..
\begin{theorem} \label{th:SDoFdeterm1}
Consider the Gaussian compound MIMO-BCC with $M$ transmit
antennas, $N_k$ receive antennas and $J_k$ channel states at
receiver $k$ for $k=1,2$. If $J_1 N_1 <M $ and $J_2N_2<M$, an
achievable s.d.o.f. region is a union of $(r_0,r_1,r_2)$ that
satisfies
\begin{flalign}
r_1 & \leq  \min(N_1,M-J_2N_2) \nonumber \\
r_2 & \leq \min(N_2,M-J_1N_1) \nonumber \\
r_0+r_1 & \leq N_1 \nonumber \\
r_0+r_2 & \leq N_2 \label{eq:sdofdeterm1}
\end{flalign}
\end{theorem}
\vspace{0.5em}

\begin{proof}(Outline)
% Gaussian superposition coding + ZF
We apply a simple linear beamforming strategy to provide an
achievable s.d.o.f. region. We first prove a useful lemma.
\begin{lemma}\label{lemma:rank}
For $0 \leq r_1 \leq \min(N_1,M-J_2N_2)$ and $0 \leq r_2 \leq
\min(N_2,M-J_1N_1)$, there exist $\vv_k^1, \ldots,\vv_k^{r_k}$ for
$k=1,2$, each with dimension $M$ that form a matrix $\Vm_k=[
\vv_k^1 \cdots \vv_k^{r_k}]$, such that
\begin{equation}\label{eq:Orth}
\Hm_{k}^j\Vm_{k'} = 0 \quad \text{for } k'\neq k, \quad
j=1,\ldots,J_{k}
\end{equation}
and $ \rank(\Hm_{k}^j\Vm_k) = r_k $ for $ j=1,\ldots,J_k$.
\end{lemma}

The proof of Lemma \ref{lemma:rank} is omitted due to space
limitations. Based on $\Vm_1$ and $\Vm_2$ given in Lemma
\ref{lemma:rank}, for the given $0 \leq r_1 \leq
\min(N_1,M-J_2N_2)$ and $0 \leq r_2 \leq \min(N_2,M-J_1N_1)$, we
let $\vv_0^1,\ldots,\vv_0^{K}$ be orthnormal vectors in the null
space of $[\Vm_1,\Vm_2]$, where $K=M-\rank[\Vm_1 \Vm_2]$. Hence,
if we let $\Vm_0=[\vv_0^1,\ldots,\vv_0^K]$, then $ \Vm_0^H
[\Vm_1,\Vm_2]= \zerov$.

We form the transmit vector at each channel use by Gaussian
superposition coding
\begin{equation}\label{GSP}
    \xv = \Vm_0 \uv_0+ \Vm_1 \uv_1 + \Vm_2 \uv_2
\end{equation}
where $\uv_0,\uv_1,\uv_2$ are mutually independent with i.i.d.
entries $u_{k,i}\sim\Nc_{\Cc}(0,p_{k,i})$ for any $k,i$ with
$u_{k,i}$ denoting the $i$-th element of $\uv_k$.

From \eqref{eq:Orth}, the received signals are given by
\begin{flalign}
    \yv_1^j & = \Hm_1^j (\Vm_0\uv_0 + \Vm_1 \uv_1) +\nv^j_1, \;\; j=1,\dots,J_1 \\
    \yv_2^l & = \Hm_2^l (\Vm_0\uv_0 + \Vm_2 \uv_2) +\nv^l_2, \;\; l=1,\dots,J_2
\end{flalign}
By letting $U=\Vm_0\uv_0$, $V_k = U+\Vm_k \uv_k$, $X=V_1+V_2$, we
obtain
\begin{small}
\begin{flalign}
&I(U; Y_{k,j})  =  \log \frac{|\Id+  \Hm_k^j (\Vm_0\diag(\pv_0)\Vm_0^H + \Vm_k \diag(\pv_k) \Vm_k^H){\Hm_k^j}^H |}{ |\Id+ \Hm_k^j \Vm_k \diag(\pv_k) \Vm_k^H  {\Hm_k^j}^H|} \label{eq:uy}\\
&I(V_k;Y_{k,j}|U)= \log|\Id + \Hm_k^j \Vm_k \diag(\pv_k) \Vm_k^H
{\Hm_k^j}^H| \label{eq:vyu}
\end{flalign}
\end{small}
In order to find the s.d.o.f., we consider equal power allocation
over all beamforming directions. we first notice that the pre-log
factor of $\log|\Id + P \Am|$ is determined by $\rank(\Am)$ as
$P\rightarrow\infty$.

From Lemma \ref{lemma:rank}, we obtain
\begin{flalign}
& \rank(\Hm_1^j\Vm_1 \Vm_1^H {\Hm_1^j}^H)=\rank(\Hm_1^j\Vm_1)= r_1
\nonumber \\
& \rank(\Hm_2^j\Vm_2 \Vm_2^H {\Hm_2^j}^H)=r_2. \nonumber
\end{flalign}
We then obtain
\begin{small}
\begin{flalign}
r_0 & = \rank \left(\Hm_1^j(\Vm_0 \Vm_0^H+\Vm_1\Vm_1^H) {\Hm_1^j}^H \right)-\rank(\Hm_1^j\Vm_1\Vm_1^H {\Hm_1^j}^H) \nonumber \\
& = N_1- r_1  \nonumber
\end{flalign}
\end{small}
and similarly, $ r_0 = N_2- r_2 $, which concludes the proof.
\end{proof}

By using the beamforming scheme similar to that for Theorem
\ref{th:SDoFdeterm1}, we obtain the following corollaries.
\begin{cor}
For the Gaussian compound MIMO-BCC with $J_1 N_1 <M $ and
$J_2N_2\geq M$, an achievable s.d.o.f. region includes
$(r_0,r_1,r_2)$ that satisfies $r_1=0$, $r_2 \leq \min(N_2,
M-J_1N_1)$, and $r_0 \leq \min(N_1,N_2-r_2)$.
\end{cor}
\begin{cor}
For the Gaussian compound MIMO-BCC with $J_1 N_1 \geq M $ and
$J_2N_2\geq M$, an achievable s.d.o.f. region includes $(r_0,0,0)$
with $r_0$ satisfying $r_0\leq \min(M,N_1,N_2)$.
\end{cor}

To gain insight into these results, we consider some special
cases. For the case of perfect CSIT ($J_1=J_2=1$), the optimal
strategy in the high SNR regime is to transmit the confidential
message $k$ in the null space of the channel matrix of the other
$k'$. This yields the s.d.o.f. $r_1\leq \min(N_1,M-N_2), r_2\leq
\min(N_2,M-N_1)$ for $M>\max(N_1,N_2)$. Clearly, the s.d.o.f. of
user $k$ corresponds to the s.d.o.f. of the MIMO wiretap channel
\cite{khisti:gmw} where the transmitter sends one confidential
message to receiver $k$ in the presence of an eavesdropper (user
$k'\neq k$). In addition, Theorem \ref{th:SDoFdeterm1} states that
if $J_1=J_2=1$ and the total number of receive antennas is large,
i.e., $N_1+N_2>M$, we can achieve the sum d.o.f. $M$. This is
certainly optimal since the MIMO-BC achieves the sum d.o.f. of
$\min(M, N_1+N_2)=M$.

For the case with a single receive antenna at each receiver, i.e.,
$N_1=N_2=1$, and without common message, i.e., $r_0=0$, we have
$r_1\leq \min(1,M-J_2)$ and $r_2\leq \min(1,M-J_1)$. This result
can be compared to the d.o.f. of the compound MIMO-BC
\cite{weingarten2007cmb}. A positive s.d.o.f. tuple $(1,1)$ is
achievable if $J_1<M$ and $J_2<M$. If the channel uncertainty of
one user increases, for example, $J_2=M$, the s.d.o.f. of user 1
collapses. Moreover, the s.d.o.f. of both users becomes zero if
$J_1\geq M, J_2\geq M$. We remark that secrecy constraints
significantly reduce d.o.f. and sometimes may yield pessimistic
results with respect to \cite{weingarten2007cmb}.

%%%%%%%%%%%%%%%%%%%%%%%%%%%%%%%%%%%%%%%%%%%%%%%%%%%%%%%%%%
\section{Ergodic Fading Compound MISO-BCC}\label{sec:ergodic}
\subsection{Model and Definitions}
We consider the MISO-BCC, where the transmitter with $M$ antennas
sends the confidential messages $W_1,W_2$ respectively to two
receivers, each equipped with single antenna. We consider the
ergodic block fading model, in which the channel remains in one
state for a block of $T$ channel uses and changes to another
channel state independently from one block to another. We assume
that the fading process is stationary and ergodic over time.
Hence, the channel state at block $t$ is given by the set of
random variables $(A_1[t],A_2[t],H[t])\in \Ac$, where
$\Ac=\{1,\dots,J_1\}\times \{1,\dots,J_2\}\times\{1,\dots,N\} $
denotes the space of fading states and each random variable is
uniformly distributed over its set. Under non-perfect CSIT, the
transmitter is assumed to know $H[t]$ and $J_1J_2$ possible states
at block $t$ but not the realization of $A_1[t]$ and $A_2[t]$, and
receiver $k$ is assumed to know both $H[t]$ and $A_k[t]$. At each
block $t$, the channel of user $k$ is expressed by two random
vectors $\hv_k^{A_k[t]}[ H[t]]$, for which we denote
$\hv_k^{A_k[t]}[t]$ for the notational simplicity. Finally, we
assume that for each $t$, any $M$ vectors taken from
$\{\hv_1^1[t],\dots,\hv_1^{J_1}[t],\hv_2^1[t],\dots,\hv_2^{J_2}[t]\}$
has rank $M$.

For each channel use at block $t$, the ergodic fading compound
MISO-BCC is expressed by
\begin{eqnarray*}\label{BlockModel}
y_k[t]= {\hv_k^j[t]}^H\xv [t] + \nu_k[t],\;\mbox{\small w.p. $P(A_k[t]=j|H[t])=\frac{1}{J_K},\forall j$ }
\end{eqnarray*}
for $k=1,2$ and $t=1,\dots,m$, where w.p. denotes with
probability, $\nu_k[t]\sim\Nc_{\Cc}(0,1)$ is an AWGN and i.i.d.
over $k,t$, and the input covariance $\Sm_x[t]$ of $\xv[t]$
satisfies the long-term power constraint $\frac{1}{m}\sum_{t=1}^m
\trace(\Sm_x[t])\leq P$. We let $n=mT$ denote the total number of
symbols over $m$ blocks. The definition for the s.d.o.f. is the
same as that in Section \ref{sec:gaussmodel}.

%%%%%%%%%%%%%%%%%%%%%%%%%%%%%%%%%%%%%%%%%%%%%%%%%%%%%%%%%%%%%%%%%%%%%
\subsection{Variable-Rate Transmission}\label{subsect:VariableRate}

We first note that as $m\rightarrow \infty, T\rightarrow\infty$,
the \emph{ergodic} achievable secrecy rate region is given by the
union of all $(R_1,R_2)$ such that \cite{LiuPoorIT09},
%\begin{eqnarray}
%  \overline{R}_s(P) = \cov \bigcup \left\{\Rm : R_k\leq \EE[I(U_k;Y_k)]-\EE[I(U_k;U_k')]-\EE[I(U_k;Y_{k'}|U_{k'})], k=1,2, k'\neq k \right\}
%\end{eqnarray}
%\begin{small}
\begin{eqnarray}\nonumber
 %\overline{R}_s(P) = \cov \bigcup \{\Rm :
& 0 \leq R_1 \leq \EE[I(V_1;Y_1)]-\EE[I(V_1;Y_{2},V_2)]\\ \label{ErgodicRegion}
& 0\leq R_2 \leq \EE[I(V_2;Y_2)]-\EE[I(V_2;Y_{1},V_{1})]
\end{eqnarray}
%\end{small}
where the expectation is with regard to the fading space $\Ac$ and
the union is over all possible distributions $V_1,V_2,X$
satisfying
\begin{equation}\label{Markov2}
  (V_1,V_2) \rightarrow X \rightarrow (Y_1, Y_2).
\end{equation}
It can be seen that the ergodic secrecy rate of user $k$ can be expressed by
\begin{eqnarray}
R_k\leq \EE[I(V_k;Y_k)]-\EE[I(V_k;V_{k'})] -
\EE[I(V_k;Y_{k'}|V_{k'})]
\end{eqnarray}
where the first two terms can be interpreted as the ergodic Marton broadcast rate without secrecy constraint, and the last term $\EE[I(V_k;Y_{k'}|V_{k'})]$ represents the information accumulated at the non-intended receiver $k'$.

We next adapt the variable-rate transmission proposed in \cite[Theorem 2]{gopala2006scf} to the compound MISO-BCC.
We focus on the zero-forcing beamforming to provide an achievable s.d.o.f. region. At each channel use of block $t$, the transmitter forms the codeword
\begin{equation}\label{ZF}
\xv[t]=\xv_1[t] + \xv_2[t]= \vv_1[t] u_1[t] + \vv_2[t] u_2[t]
\end{equation}
where $\vv_k[t]$ denotes a unit-norm beamforming vector of user $k$ (to be specified below) and $u_k[t]\sim\Nc_{\Cc}(0,p_k[t])$ is symbol of user $k$, and $u_1[t],u_2[t]$ are mutually independent.
Clearly, the Markov chain (\ref{Markov2}) is satisfied by letting $V_k=\vv_k[t] u_k[t]$ and $X=V_1+V_2$ at each $t$.
Following \cite{gopala2006scf}, we assume that the transmitter sends the codeword $\xv_k[t]$ to user $k$ at rate given by
\begin{eqnarray}\label{TxRate}
R_{k,\rm tx}[t] &= & I(u_k[t];y_k[t])-I(u_k[t];u_{k'}[t]) \\
\nonumber
 &\overset{\mathrm{(a)}} {=}& I(u_k[t];y_k[t]) \\  \nonumber
&=& \sum_{j=1}^{J_k}P(A_k[t]=j|H[t]) I(u_k[t];y_k[t]|A_k[t]=j)
\end{eqnarray}
where (a) follows from the independency between $u_1[t]$ and
$u_2[t]$. This variable-rate strategy enables to limit the leaked
information at the non-intended receiver $k'$ at each block $t$
such that
\begin{small}
\begin{eqnarray}\nonumber
\lefteqn{I(u_k[t]; y_{k'}[t]|u_{k'}[t])}\\ \nonumber
&=& \sum_{j=1}^{J_{k'}}P(A_{k'}[t]=j|H[t])I(u_k[t]; y_{k'}[t]|u_{k'}[t],A_{k'}[t]=j)
\\ \label{VariableRateUB}
&\leq& R_{k,\rm tx}[t]
\end{eqnarray}
\end{small}
for $k'\neq k$ and $k=1,2$.
By combining (\ref{TxRate}) and (\ref{VariableRateUB}), the averaged secrecy rate of user $k$ over $m$ blocks is given by
\begin{eqnarray}\nonumber
    R_k^m &=& \frac{1}{m}\sum_{t=1}^m R_{k,\rm tx}[t]- \frac{1}{m}\sum_{t=1}^m I(u_k[t]; y_{k'}[t]|u_{k'}[t])\\ \label{AveragedSecRate}
&=& \frac{1}{m}\sum_{t=1}^m R_k[t]
\end{eqnarray}
where the secrecy rate of user $k$ at block $t$ is given by
\begin{eqnarray}\label{SecRateBlock}
R_k[t] &=&  [R_{k,\rm tx}[t]-I(u_k[t]; y_{k'}[t]|u_{k'}[t])]_+
\end{eqnarray}
We remark that similar to \cite{gopala2006scf}, the variable rate strategy avoids the non-intended receiver $k'$ to accumulate the information on symbol $k$ over $m$ blocks, whenever the channel condition is better than the transmission rate of user $k$.

\subsection{Secrecy Degree of Freedom Region}
In the following, we provide the s.d.o.f. analysis for different
cases of $(J_1,J_2)$.

\begin{theorem} The two-user ergodic fading compound MISO-BCC with $J_1<M,J_2<M$ achieves the s.d.o.f. region
\[
\{(r_1,r_2): r_1\leq 1, r_2\leq 1 \} . \]
\end{theorem}
\begin{proof}
At each block $t$, the transmitter forms $\xv[t]$ given in (\ref{ZF}) by choosing $\vv_1[t]$ orthogonal to $\hv_2^1[t],\dots,\hv_2^{J_2}[t]$ and $\vv_2[t]$ orthogonal to $\hv_1^1[t],\dots,\hv_1^{J_1}[t]$. This
yields the received signals for $k=1,2$ given by
\begin{eqnarray*}
y_k[t] =\phi_{k,k}^j[t]u_k[t]+ \nu_k[t] , \mbox{\small w.p. $P(A_k[t]=j|H[t])=\frac{1}{J_k},\forall j$ }
\end{eqnarray*}
where $\phi_{k,i}^j[t]={\hv_k^j[t]}^H \vv_i[t]$. It can be shown
that $\vv_1[t]$ and $\vv_2[t]$ can be chosen such that
$\phi_{k,k}^j[t] \neq 0$.
Since the ZF creates two parallel channels for any pair $(A_1[t],A_2[t])$, the averaged secrecy rate of user $k$ over $m$ blocks is readily given by
\begin{eqnarray*}
 R_k^m %&\leq & \frac{1}{m J_k}\sum_{t=1}^m [R_{k,\rm tx}[t]- I(u_k[t];y_k[t]|u_{k'}[t]) ]\\
 &\leq & \frac{1}{m J_k}\sum_{t=1}^m  \sum_{j=1}^{J_k} C(p_k[t]|\phi_{k,k}^j[t]|^2) %\\
 %R_2 &=& \frac{1}{m J_2}\sum_{t=1}^m  \sum_{j=1}^{J_2} \log(1+p_2[t]|{\hv_2^j[t]}^H \vv_2[t]|^2)
\end{eqnarray*}
As $m\rightarrow\infty$, the corner point $(1,0),(0,1)$ is
achieved by allocating $p_1[t]=P,\forall t$, $p_2[t]=P, \forall
t$, respectively, and the rate point $(1,1)$ is achieved by equal
power allocation $p_1[t]=p_2[t]=P/2$ at each $t$. Time-sharing
between three points yields the region.
\end{proof}
\begin{theorem}
The two-user ergodic fading compound MISO-BCC with $J_1<M,J_2\geq
M$ achieves the s.d.o.f. region (see
Fig.\ref{fig:SDFregionMcase12}) that includes $(r_1,r_2)$
satisfying
\begin{eqnarray}
 r_1\leq \frac{M-1}{J_2}, \quad\quad \left(\frac{J_2}{M-1}-1\right)r_1 + r_2 \leq 1
\end{eqnarray}
\end{theorem}
\begin{proof}
At each block $t$, the transmitter chooses $\vv_1[t]$ orthogonal
to the first $M-1$ states\footnote{The same result holds for any
$M-1$ set taken from $\{\hv_2^1[t],\dots,\hv_2^{J_2}[t]\}$. }
$\hv_2^1[t],\dots,\hv_2^{M-1}[t]$ and $\vv_2[t]$ orthogonal to
$\hv_1^1[t],\dots,\hv_1^{J_1}[t]$ to form the codeword (\ref{ZF})
at each $t$. This yields the receive signals
\begin{small}
\begin{eqnarray*}
&y_1[t] =
\begin{array}
[c]{ll}%
 \phi_{1,1}^j[t] u_1[t]+ \nu_1[t] , &\mbox{w.p. $P(A_1[t]=j|H[t])=\frac{1}{J_1},\forall j$}\\
\end{array}\\
&y_2[t] =\left\{
\begin{array}
[c]{l}%
\phi_{2,2}^j[t]u_2[t]+ \nu_2[t] , \\
 \quad\quad\mbox{w.p. $P(A_2[t]=j|H[t])=\frac{1}{J_2}$ for $j\leq M-1$}\\
\phi_{2,1}^j[t]u_1[t] +\phi_{2,2}^j[t]u_2[t]+ \nu_2[t] \\
 \quad\quad\mbox{w.p. $P(A_2[t]=j|H[t])=\frac{1}{J_2}$ for $M\leq j\leq J_2$}\\
\end{array}
\right.
\end{eqnarray*}
\end{small}
{\normalsize We remark that the increased channel uncertainty at
user 2 ($J_2\geq M$) incurs two effects. First, it decreases the
transmission rate of user 2 due to interference from user 1.
Second, it decreases the secrecy rate of user 1 since user 2
observes $u_1[t]$ with probability $\frac{J_2-M+1}{J_2}$, if
$A_k[t]$ is between $M$ and $J_2$. We obtain the secrecy rates at
block $t$ given by}
\begin{footnotesize}
\begin{flalign}
R_1 [t]&\leq  \left[\frac{1}{J_1}\sum_{j=1}^{J_1}C(p_1[t]|\phi_{1,1}^j[t]|^2)-\frac{1}{J_2}\sum_{j=M}^{J_2} C( p_1[t]|\phi_{2,1}^j[t]|^2 )\right]_+ \nonumber\\
%R_2[t] \leq  \sum_{j=1}^{M-1}\log(1+ p_2[t]|\phi_{2,2}^j[t]|^2)+ \sum_{j=M}^{J_2}
%\log\left(1 + \frac{p_2[t]|\phi_{2,2}^j[t]|^2}{1+ p_1[t]| \phi_{2,1}^j[t]|^2}\right)
R_2[t] &\leq
\frac{1}{J_2}\sum_{j=1}^{M-1}C(p_2[t]|\phi_{2,2}^j[t]|^2)+
\frac{1}{J_2}\sum_{j=M}^{J_2}
C\left(\frac{p_2[t]|\phi_{2,2}^j[t]|^2}{1+ p_1[t]|
\phi_{2,1}^j[t]|^2}\right) \nonumber
\end{flalign}
\end{footnotesize}
Plugging these expressions into (\ref{AveragedSecRate}) and letting $m\rightarrow\infty$, the corner point $(0,1), (\frac{M-1}{J_2},0)$ is achieved by letting $p_2[t]=P$ and $ p_1[t]=P$ for all $t$. Under equal power allocation $p_1[t]=p_2[t]=\frac{P}{2}$ for all $t$, $\left(\frac{M-1}{J_2},\frac{M-1}{J_2}\right)$ is achieved. Time-sharing of these three points yields the region.
\end{proof}

\begin{figure}[t]
    \begin{center}
\epsfxsize=1.7in
\epsffile{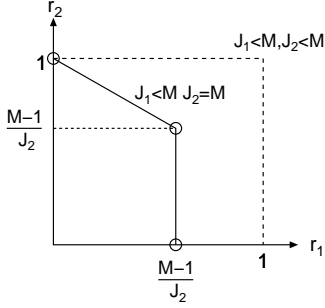}
    \end{center}
    \vspace{-0.2in}
    \caption{s.d.o.f. region for $J_1,J_2<M$ and $J_1<M, J_2\geq M$}
    \label{fig:SDFregionMcase12}
\end{figure}
\begin{theorem}
Consider the two-user ergodic compound MISO-BCC with $J_1\geq
M,J_2\geq M$. We define the function
$f(J_1,J_2)=\frac{M-1}{J_1}+\frac{M-1}{J_2}-1-\frac{M-1}{J_1+J_2}$.
If $f(J_1,J_2)\leq 0$, an achievable region is given by the
time-sharing between $(\frac{M-1}{J_2},0)$ and
$(0,\frac{M-1}{J_1})$. If $f(J_1,J_2)> 0$, an achievable region
(see Fig.\ref{fig:SDFregionMcase3}) is time-sharing between these
two points and $(r_s,r_s)$ with
$r_s=\frac{M-1}{J_1}+\frac{M-1}{J_2}-1$.
\end{theorem}
\begin{proof}
Without loss of generality, the transmitter chooses $\vv_1[t]$
orthogonal to $\hv_2^1[t],\dots,\hv_2^{M-1}[t]$ and $\vv_2[t]$
orthogonal to $\hv_1^1[t],\dots,\hv_1^{M-1}[t]$ to form the
codeword given in (\ref{ZF}) at block $t$. This yields the receive
signals
\begin{small}
\begin{eqnarray*}
y_k[t] =\left\{
\begin{array}
[c]{l}
\phi_{k,k}^j[t] u_k[t]+ \nu_k[t] , \\
  \quad\quad\mbox{w.p. $P(A_k[t]=j|H[t])=\frac{1}{J_k}$ for $j\leq M-1$}\\
 \phi_{k,k}^j[t] u_k[t]+ \phi_{k,k'}^j[t]u_{k'}[t]+\nu_k[t] , \\
 \quad\quad\mbox{w.p. $P(A_k[t]=j|H[t])=\frac{1}{J_k}$ for $M\leq j\leq J_k$}\\
\end{array}\right. %\\
%y_2[t] &=&\left\{
%\begin{array}
%[c]{ll}%
% {\hv_2^j[t]}^H \vv_2[t]u_2[t]+ \nu_2[t] , & \mbox{w.p. $P(A_2[t]=j|H[t])=\frac{1}{J_2}$ for $j\leq M-1$}\\
%{\hv_2^j[t]}^H \vv_1[t] u_1[t] +{\hv_2^j[t]}^H \vv_2[t]u_2[t]+ \nu_2[t] , & \mbox{w.p. $P(A_2[t]=j|H[t])=\frac{1}{J_2}$ for $M \leq j\leq J_2$}\\
%\end{array}
%\right.
\end{eqnarray*}
\end{small}
for $k=1,2$.
By taking into account the two effects caused by the increased channel uncertainty mentioned above,
we obtain the secrecy rate of user $k$ at block $t$ is given by
\begin{footnotesize}
\begin{eqnarray*}
\begin{split}
R_k[t] = \left[\frac{1}{J_k}\sum_{j=1}^{M-1} C( p_k[t]|\phi_{k,k}^j[t]|^2) -\frac{1}{J_{k'}}\sum_{j=M}^{J_{k'}}C( p_k[t] |\phi_{k',k}^j[t]|^2 )\right. \\ \nonumber
+
\left.\frac{1}{J_k}\sum_{j=M}^{J_k}C\left(\frac{p_k[t]|\phi_{k,k}^j[t]|^2}{1+p_{k'}[t]
 |\phi_{k,k'}^j[t]|^2}\right)\right]_+
\end{split}
\end{eqnarray*}
\end{footnotesize}
for $k=1,2$.
Plugging the above expression into (\ref{AveragedSecRate}) and letting $m\rightarrow\infty$,
the corner point $A=\left(\frac{M-1}{J_2},0\right), B=\left(0,\frac{M-1}{J_1}\right)$ is achieved by letting $p_1[t]=P, p_2[t]=P, \forall t$, respectively. Under equal power allocation $p_1[t]=p_2[t]=P/2,\forall t$,
we have two different behaviors according to the value of $f(J_1,J_2)$. Interestingly, if $f(J_1,J_2)>0$, the s.d.o.f. point $C=(r_s,r_s)$ which dominates the line segment A B is achieved, as shown in Fig.\ref{fig:SDFregionMcase3}. On the contrary, if $f(J_1,J_2)\leq 0$, the point $(r_s,r_s)$ is below the line segment A B. This can be easily verified by comparing $r_s$ and the intersection between the line segment A B and $r_2=r_1$.
\end{proof}

\begin{figure}[t]
    \begin{center}
\epsfxsize=3.2in
\epsffile{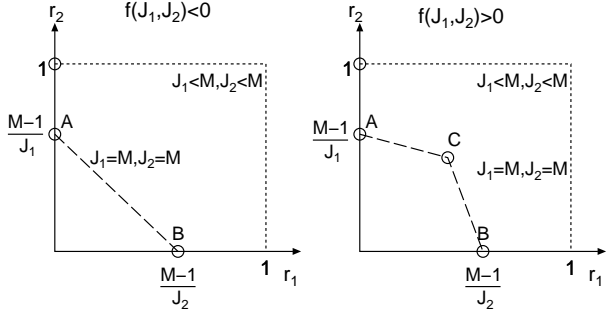}
    \end{center}
    \vspace{-0.2in}
    \caption{s.d.o.f. region for $J_1\geq M, J_2\geq M$}
    \label{fig:SDFregionMcase3}
    \vspace{-0.15in}
\end{figure}

We remark that an achievable s.d.o.f. with the ergodic model
gradually decreases as the uncertainty increases. Moreover, the
time variation of the channel state creates an additional temporal
dimension, and significantly improves the s.d.o.f. with respect to
the Gaussian model with constant channel state. We provide a
simple example to illustrate the difference between two models.
Consider the compound MISO-BCC with $M=7, J_1=J_2=8$. The ergodic
model achieves $(1/2, 1/2)$ which dominates the time-sharing
between the corner points $(3/4,0)$ and $(0,3/4)$. The Gaussian
model yields zero s.d.o.f. for both users. This radical difference
is because the number of channel states over which perfect secrecy
must be kept for the Gaussian model equals the number of
wiretappers, which is not the case for the ergodic model.

%%%%%%%%%%%%%%%%%%%%%%%%%%%%%%%%%%%%%%%%%%%%%%%%%%%%%%%%%%
%\vspace{-0.5em}
\section{Conclusions}\label{sec:conclusions}
We have studied the two-user compound MIMO-BCC, for which we have
found that time variation of the channel state provides an
additional temporal dimension for the ergodic model, which
improves an achievable s.d.o.f. region compared to the Gaussian
model with a constant fading state, although at the price of a
larger delay. We note also that in contrast to the compound
MIMO-BC \cite{weingarten2007cmb}, the gain by multiletter
approaches (i.e. combining several time instances) is not expected
here. Finally, we conjecture that an achievable s.d.o.f. region
provided in the paper is indeed the s.d.o.f. region and the proof
remains as a future investigation.

%\vspace{-0.5em}
\section*{Acknowledgment}
This work is partially supported by the European Commission in the framework of the FP7 Network of Excellence in Wireless Communications NEWCOM++.

\bibliographystyle{IEEEtran}
\bibliography{CompoundMIMOBCC-CR}

%\nopagebreak[4]
\end{document}

%% file: macros.tex
\newfont{\bbb}{msbm10 scaled 500}

\newfont{\bb}{msbm10 scaled 1100}
\newcommand{\CC}{\mbox{\bb C}}

\newcommand{\EE}{\mbox{\bb E}}

\newcommand{\hv}{{\bf h}}

\newcommand{\nv}{{\bf n}}

\newcommand{\pv}{{\bf p}}

\newcommand{\uv}{{\bf u}}

\newcommand{\vv}{{\bf v}}
\newcommand{\xv}{{\bf x}}
\newcommand{\yv}{{\bf y}}

\newcommand{\zerov}{{\bf 0}}

\newcommand{\Am}{{\bf A}}

\newcommand{\Hm}{{\bf H}}
\newcommand{\Id}{{\bf I}}

\newcommand{\Sm}{{\bf S}}

\newcommand{\Vm}{{\bf V}}

\newcommand{\Ac}{{\cal A}}

\newcommand{\Cc}{{\cal C}}

\newcommand{\Nc}{{\cal N}}

\newcommand{\Wc}{{\cal W}}

\newcommand{\nuv}{\hbox{\boldmath$\nu$}}

\newcommand{\diag}{{\hbox{diag}}}

\newcommand{\trace}{{\hbox{tr}}}
\newcommand{\rank}{{\hbox{rank}}}